\documentclass[letterpaper, 10 pt, conference]{ieeeconf}  % Comment this line out
\IEEEoverridecommandlockouts                              % This command is only
\usepackage{amsmath} % assumes amsmath package installed
\usepackage{amssymb}  % assumes amsmath package installed
\usepackage{url}
\usepackage{graphicx}
\usepackage{graphics,subfigure} % for pdf, bitmapped graphics files

\newtheorem{corollary}{\textbf{Corollary}}
\newtheorem{lemma}{\textbf{Lemma}}
\newtheorem{theorem}{\textbf{Theorem}}
\newtheorem{definition}{\textbf{Definition}}
\newtheorem{example}{\textbf{Example}}
\newtheorem{algorithm}{\textbf{Algorithm}}
\newtheorem{experiment}{\textbf{Experiment}}

\title{\LARGE \bf
Generalized Network Tomography
}

\author{\authorblockN{Gugan Thoppe}
\authorblockA{School of Technology and Computer Science\\
Tata Institute of Fundamental Research, Mumbai, INDIA\\
Email: gugan@tcs.tifr.res.in}}

\begin{document}

\maketitle
\thispagestyle{empty}
\pagestyle{empty}

%%%%%%%%%%%%%%%%%%%%%%%%%%%%%%%%%%%%%%%%%%%%%%%%%%%%%%%%%%%%%%%%%%%%%%%%%%%%%%%%
\begin{abstract}
For successful estimation, the usual network tomography algorithms crucially require i) end-to-end data generated using multicast probe packets, real or emulated, and ii) the network to be a tree rooted at a single sender with destinations at leaves. These requirements, consequently, limit their scope of application. In this paper, we address successfully a general problem, henceforth called generalized network tomography, wherein the objective is to estimate the link performance parameters for networks with arbitrary topologies using only end-to-end measurements of pure unicast probe packets. Mathematically, given a binary matrix $A,$ we propose a novel algorithm to uniquely estimate the distribution of $X,$ a vector of independent non-negative random variables, using only IID samples of the components of the random vector $Y = AX.$ This algorithm, in fact, does not even require any prior knowledge of the unknown distributions. The idea is to approximate the distribution of each component of $X$ using linear combinations of known exponential bases and estimate the unknown weights. These weights are obtained by solving a set of polynomial systems based on the moment generating function of the components of  $Y.$ For unique identifiability, it is only required that every pair of columns of the matrix $A$ be linearly independent, a property that holds true for the routing matrices of all multicast tree networks. Matlab based simulations have been included to illustrate the potential of the proposed scheme.
\end{abstract}

%%%%%%%%%%%%%%%%%%%%%%%%%%%%%%%%%%%%%%%%%%%%%%%%%%%%%%%%%%%%%%%%%%%%%%%%%%%%%%%%
\section{Introduction}
\label{sec:intro}
Network tomography, first proposed in \cite{Vardi96}, is the science of inferring spatially localized network behavior using only metrics that are practically feasible to measure. The problems considered in network tomography can be classified into two broad strands: i) traffic demand tomography---determination of source destination traffic volumes via measurements of link volumes and ii) network delay tomography---link parameter estimation based on end to end path level measurements. However, central to both these areas, is the problem of inferring the statistics of $X,$ a vector of independent non-negative random variables, given the measurement model $Y = AX.$ The challenge in these problems stems from the fact that $A$ is usually an ill posed matrix and hence non-invertible. For excellent tutorials and surveys on the state of art, see \cite{Adams00,Coates02b,Castro04} and \cite{Lawrence07}.

For sake of definiteness, we consider here the problem of network delay tomography. The proposed method, however, is also applicable to traffic demand tomography. Under delay tomography, the major problems studied include estimation of bottleneck link bandwidths, e.g. \cite{Liu07,Dey11}, link loss rates, e.g. \cite{Caceres99}, link delays, e.g., \cite{Coates01,LoPresti02,Shih03,Tsang03}, etc. For successful estimation, the proposed solutions to these problems crucially require i) end-to-end data generated using multicast probe packets, real or emulated, and ii) the network to be a tree rooted at a single sender with destinations at leaves. These algorithms mainly exploit the correlations in the path measurements, i.e., the packets have the same experience on shared links. Because of this, any divergence in either of the above requirements results in performance degradation. Consequently, there is a need to develop tomography algorithms for networks with arbitrary topologies using only pure unicast probe packet measurements. Mathematically, this is same as addressing the generalized network tomography (GNT) problem, wherein, given the binary matrix $A,$ the objective is to estimate the statistics of $X,$ a vector of independent non-negative random variables, using only IID samples of the components of the random vector $Y = AX.$ 

In this paper we propose a novel method, henceforth called the distribution tomography (DT) scheme, for the framework of GNT to accurately estimate the distribution of $X$ even when no prior knowledge about the same is available. We rely on the fact that the class of \textit{generalized hyperexponential} (GH) distributions is dense in the set of non-negative distributions (see \cite{Botha86}). Using this, the idea then is to approximate the distribution of each component of $X$ using linear combinations of known exponential bases and estimate the unknown weights. These weights are obtained by solving a set of polynomial systems based on the moment generating function of the components of  $Y.$ For unique identifiability, it is only required that every pair of columns of the matrix $A$ be linearly independent, a property that holds true for the routing matrices of all multicast tree networks.

The rest of the paper is organized as follows. In the next section, we develop the notation and formally describe the problem. Section~\ref{sec:approxDist} recaps the theory of approximating non-negative distributions using linear combinations of exponentials. In Sections~\ref{sec:GenNT} and \ref{sec:universality}, we develop our proposed method and demonstrate its universal applicability. We give numerical examples in Section~\ref{sec:ExptResults} and end with a short discussion in Section~\ref{sec:ConclFutW}.

\section{Model and Problem Description}
\label{sec:model}
Any cumulative distribution function (CDF) that we work with is always assumed to be continuous with support $(0, \infty).$ The moment generating function (MGF) of the random variable $X$ will be $M_X(t) = \mathbb{E}(\exp(-tX)).$ For $n \in \mathbb{N},$ we use $[n]$ and $S_n$ to represent respectively the set $\{1, \ldots, n\}$ and its permutation group. We use the notation $\mathbb{R}, \mathbb{R}_{+}$ and $\mathbb{R}_{++}$ to denote respectively the set of real numbers, non-negative real numbers and strictly positive real numbers. In the same spirit, for integers we use $\mathbb{Z}, \mathbb{Z}_{+}$ and $\mathbb{Z}_{++}.$ All vectors are column vectors and their lengths refer to the usual Euclidean norm. For $\delta > 0,$ $B(\mathbf{v}; \delta)$ represents the open $\delta-$ball around the vector $\mathbf{v}.$ To denote the derivative of the map $f$ with respect to $\mathbf{x},$ we use $\dot{f}(\mathbf{x}).$ Lastly, all empty sums and empty products equal $0$ and $1$ respectively.

Let $X_1, \ldots, X_N$ denote the independent non-negative random variables whose distribution we wish to estimate. We assume that each $X_j$ has a GH distribution of the form
\begin{equation}
  \begin{array}{cccc}
    F_j(u) & = & \sum_{k=1}^{d + 1} w_{jk} \left[ 1 - \exp\left( -\lambda_{k} u \right) \right],  & u \geq 0
  \end{array}
	\label{eq:Dist}
\end{equation}
where $\lambda_{k}>0,$ $\sum_{k=1}^{d+1} w_{jk} \lambda_{k} \exp(-\lambda_{k}u) \geq 0$ and $\sum_{k=1}^{d + 1}w_{jk} =  1.$ Further, we suppose that $\lambda_1, \ldots, \lambda_{d + 1}$ are distinct and explicitly known and that the weight vectors of distinct random variables differ at least in one component. Let $A \in \{0,1\}^{m \times N}$ denote an a priori known matrix which is $1-$identifiable in the following sense.
\begin{definition}
A matrix $A$ is $k-$identifiable if every set of $2k$  of its columns is linearly independent.
\end{definition}

Let $X \equiv (X_1, \ldots, X_N)$ and $Y = AX.$ For each $i \in [m],$ we  presume that we have access to a sequence of IID samples of the $i^{th}$ linear combination $Y_i.$ Our problem then is to estimate for each $X_j,$ its vector of weights $\mathbf{w}_j \equiv (w_{j1}, \ldots, w_{jd})$ and consequently its complete distribution $F_j,$ since $w_{j(d+1)} = 1 - \sum_{k = 1}^{d} w_{jk}.$

Before developing the estimation procedure, we begin by making a case for the distribution model of \eqref{eq:Dist}.

\section{Approximating Distribution Functions}
\label{sec:approxDist}
Consider the problem of simultaneously estimating the members of a finite family of arbitrary distributions, say $\mathcal{G}=\{G^1, \ldots, G^N\}.$ A useful strategy is to approximate each member by a GH distribution. The CDF of a GH random variable $X$ is given by
\begin{equation}
  \begin{array}{cc}
    F_{X}(u) = \sum_{k=1}^{d + 1} \alpha_k \left[ 1 - \exp\left( -\lambda_{k} u \right) \right],  & u \geq 0,
  \end{array}
  \label{eq:GHDist}
\end{equation}
where $\lambda_{k}>0,$ $\sum_{k=1}^{d + 1}\alpha_{k}\lambda_{k}\exp(-\lambda_{k}u)\geq0$ and $\sum_{k=1}^{d + 1}\alpha_{k}=1.$ Consequently, its MGF is given by
\begin{equation}
M_X(t)=\sum_{k=1}^{d + 1}\alpha_{k}\frac{\lambda_{k}}{\lambda_{k} + t}. \label{eq:GHMGF}
\end{equation}

In addition to the simple algebraic form of the above quantities, the other major reason to use the GH class is that, in the sense of weak topology, it is dense in the set of distributions (see \cite{Botha86}). In fact given a continuous CDF $F$ with MGF $M$, one can explicitly find a sequence of GH distributions, say $F_n,$ that converge uniformly to it. Furthermore, if $M_n$ is the MGF for $F_n,$ then for each $t \geq 0,$ $M_n(t)\rightarrow M(t).$ The idea is based on the following result.
\begin{theorem}
	\cite{Ou97}	\label{thm:Ou}
For $n,k \in \mathbb{N},$ let $X_{n,k}$ be a nonnegative GH random variable with mean $k/n,$ variance $\sigma_{n,k}^{2}$ and CDF $W_{n,k}.$ Suppose
  \begin{enumerate}
  \item the function $\nu:\mathbb{N}\rightarrow\mathbb{N}$ satisfies $\underset{n\rightarrow\infty}{lim}\nu(n)/n=\infty.$
  \item there exists $0<s<1$ such that $\underset{n\rightarrow\infty}{lim}n^{1+s} \sigma_{n,k}^{2}/k=0$ uniformly with respect to $k.$
  \end{enumerate}
  Then given any continuous non-negative distribution function $F,$ the following holds:
  \begin{enumerate}
  \item the function $F_{n}$ given by
    \begin{eqnarray*}
      F_{n}(u) & = & \sum_{k=1}^{\nu(n)}\left\{F(k/n)-F((k-1)/n)\right\} W_{n,k}(u)\\
      &  & +(1-F(\nu(n)/n))W_{n,\nu(n)+1}(u)
    \end{eqnarray*}
    is a GH distribution for every $n\in\mathbb{N}$ and
  \item $F_{n}$ converges uniformly to $F,$ i.e., 
  \begin{equation*}
    \underset{n \rightarrow \infty}{\lim} \underset{-\infty < u < \infty}{\sup} \left|F_{n}(u)-F(u)\right| =  0.
  \end{equation*}\,
  \end{enumerate}
\end{theorem}

Observe that for each $n,$ the exponential stage parameters of $F_n$ depend only on the choice of the random variables $\{X_{n,k}\!: 1\leq k\leq \nu(n) +1\}.$ Regarding $\mathcal{G},$ if we fix the random variables $X_{n,k}$ and let $M^i$ denote the MGF of $G^i,$  then this observation and the above result implies that for any given $\epsilon_1, \epsilon_2 >0$ and any finite set $\tau=\{t_1,\ldots ,t_k\} \subset \mathbb{R}_+,$ $\exists n \equiv n(\epsilon_1, \epsilon_2, \tau) \in \mathbb{N}$ such that for each $i\in[N],$ $G^i$ and its $n^{th}$ GH approximation, $F_n^i,$ are $\epsilon_1-$ close in the sup norm and for each $j \in [k],$ $|M^i(t_j)-M_n^i(t_j)|\leq \epsilon_2.$ Further, the exponential stage parameters are explicitly known and identical across the approximations $F_n^i,$ which now justifies our model of \eqref{eq:Dist}.

The problem of estimating the members of $\mathcal{G}$ can thus be reduced to determining the vector of weights that characterizes each approximation and hence each distribution.

\section{Distribution Tomography Scheme}
\label{sec:GenNT}

The outline for this section is as follows. For each $i,$ we use the IID samples of $Y_i$ to estimate its MGF and subsequently build a polynomial system, say $H_i(\mathbf{x}) = 0.$ We call this the elementary polynomial system (EPS). We then show that for each $i \in [m]$ and each $j \in p_i := \{j \in [N]: a_{ij} = 1\},$ a close approximation of the vector $\mathbf{w}_j$ is present in the solution set of $H_i(\mathbf{x}) = 0,$ denoted $V(H_i).$ To match the weight vectors to the corresponding random variables, we make use of the fact that $A$ is $1-$identifiable.

\subsection{Construction of Elementary Polynomial Systems}
\label{subsec:ConstPolySys}

Fix $i \in [m]$ and suppose that $|p_i| = N_i,$ i.e., $Y_i$ is a sum of $N_i$ random variables, which for convenience, we relabel as $X_1, \ldots, X_{N_i}$ in some order. Using \eqref{eq:Dist} and \eqref{eq:GHMGF}, observe that the MGF of $Y_i$ is well defined $\forall t \in \mathbb{R}_{++}$ and satisfies the relation
\begin{equation}
\label{eq:MGFRel}
M_{Y_{i}}(t) = \prod_{j = 1}^{N_i} \left\{ \sum_{k = 1}^{d + 1}w_{jk} \left( \frac{\lambda_{k}}{\lambda_{k}+t} \right) \right\}.
\end{equation}
On simplification, after substituting $w_{j(d+1)} = 1 - \sum_{k=1}^{d}w_{jk},$ we get
\begin{equation}
\label{eq:modMGFRel}
\mu_i(t) = \left\{ \prod_{j = 1}^{N_i} \left[ \sum_{k=1}^{d}w_{jk}\Lambda_k(t) + \lambda_{d + 1}\right] \right\},
\end{equation}
where $\Lambda_k(t) = (\lambda_{k} - \lambda_{d + 1})t/(\lambda_{k} + t)$ and $\mu_i(t) = M_{Y_{i}}(t) (\lambda_{d + 1} + t)^{N_i} .$

For now, let us assume that we know $M_{Y_i}(t)$ and hence $\mu_i(t)$ exactly for every valid $t.$ We will refer henceforth to this situation as the ideal case. Treating $t$ as a parameter, we can then use \eqref{eq:modMGFRel} to define a canonical polynomial
\begin{equation}
f(\mathbf{x};t) = \left\{ \prod_{j = 1}^{N_i} \left[ \sum_{k=1}^{d}x_{jk}\Lambda_k(t) + \lambda_{d + 1}\right] \right\} - \mu_i(t), \label{eq:tempGenPoly}
\end{equation}
where $\mathbf{x} \equiv (\mathbf{x}_1, \ldots, \mathbf{x}_{N_i})$  with $\mathbf{x}_j \equiv (x_{j1}, \ldots, x_{jd}).$ As this is a multivariate map in $d \cdot N_i$ variables, we can choose an arbitrary set $\tau = \{t_1, \ldots, t_{d \cdot N_i}\} \subset \mathbb{R}_{++}$ consisting of distinct numbers and define an intermediate square polynomial system
\begin{equation}
\label{eq:tempPolSys}
F_\tau(\mathbf{x}) \equiv (f_1(\mathbf{x}), \ldots, f_{d \cdot N_i}(\mathbf{x})) = 0,
\end{equation}
where $f_k(\mathbf{x}) \equiv f(\mathbf{x}; t_k).$

Since \eqref{eq:tempPolSys} depends on choice of $\tau,$ analyzing it directly is difficult. But observe that i) the expansion of each $f_n$ or equivalently \eqref{eq:tempGenPoly} results in rational coefficients in $t$ of the form $\Lambda^\mathbf{L} \equiv \Lambda_{1}^{L_{1}}(t)\cdots\Lambda_{d}^{L_{d}}(t)\lambda_{d+1}^{L_{d+1}},$ where $L_k \in \mathbb{Z}_{+}$ and $\sum_{k=1}^{d+1}L_k = N_{i},$ and ii) the monomials that constitute each polynomial are identical. This suggests that one may be able to get a simpler representation for \eqref{eq:tempPolSys}. We do so in the following three steps, where the first two focus on simplifying \eqref{eq:tempGenPoly}.

\noindent \emph{Step1-Gather terms with common coefficients}: Let
\[
\Delta_{d+1,N_{i}}:=\left\{ \mathbf{L}\equiv\left(L_{1},\ldots,L_{d+1}\right)\in\mathbb{Z}_{+}^{d+1}:\sum_{k=1}^{d+1}L_{k}=N_{i}\right\}.
\]
For a vector $\mathbf{b}\equiv(b_{1},\ldots,b_{N_{i}})\in[d+1]^{N_{1}},$ let its type be denoted by $\Theta(\mathbf{b})\equiv(\theta_{1}(\mathbf{b}),\ldots,\theta_{d+1}(\mathbf{b})),$ where $\theta_{k}(\mathbf{b})$ is the count of the element $k$ in $\mathbf{b}.$ For every $\mathbf{L}\in\Delta_{d+1,N_{i}},$ additionally define the set
\[
\mathcal{B}_{\mathbf{L}}=\left\{ \mathbf{b}\equiv(b_{1},\ldots,b_{N_{i}})\in[d+1]^{N_{i}}:\,\Theta(\mathbf{b})=L\right\}
\]
and the polynomial $g(\mathbf{x};\mathbf{L})=\sum_{\mathbf{b}\in\mathcal{B}_{L}}(\prod_{j=1,\, b_{j}\ne d+1}^{N_{i}}x_{jb{}_{j}}).$ Then collecting terms with common coefficients in \eqref{eq:tempGenPoly}, the above notations help us rewrite it as
\begin{equation}
f(\mathbf{x};t)=\sum_{\mathbf{L}\in\Delta_{d+1,N_{i}}}g(\mathbf{x};\mathbf{L})\Lambda^{\mathbf{L}}-\mu_{i}(t).\label{eq:afterColComTerms}
\end{equation}

\noindent \emph{Step2-Coefficient expansion and regrouping}: Using an idea similar to partial fraction expansion for rational functions in $t$, the goal here is to decompose each coefficient into simpler terms. For each $j,$ $k\in[d],$ let
\[
\beta_{jk}:=
\begin{cases}
\frac{\lambda_{j}\left(\lambda_{k}-\lambda_{d+1}\right)}{\left(\lambda_{j}-\lambda_{k}\right)} & \mbox{if }j\ne k\\
1 & j=k
\end{cases}.
\]
For each $\mathbf{L} \in \Delta_{d + 1, N_i},$ let $\mathcal{D}(\mathbf{L}):= \{k\in[d]: L_{k}>0\}.$ Further, if $\mathcal{D}(\mathbf{L}) \neq \emptyset$ then $\forall k \in \mathcal{D}(\mathbf{L})$ and $\forall q \in [L_k],$ let
\begin{eqnarray*}
\bar{\Delta}_{kq}(\mathbf{L}) & := & \{\mathbf{s}\equiv(s_{1},\ldots,s_{d})\in\mathbb{Z}_{+}^{d}: s_{r} = 0\\
& &\forall r\in\mathcal{D}(\mathbf{L})^{c} \cup \{k\}; \sum_{n=1}^{d}s_n = L_{k}-q\} \mbox{ and}
\end{eqnarray*}
\[
\gamma_{kq}(\mathbf{L}):=\prod_{r\in\mathcal{D}(\mathbf{L})}\beta_{kr}^{L_{r}}\left\{\sum_{\mathbf{s}\in\bar{\Delta}_{kq}(\mathbf{L})}\prod_{r\in\mathcal{D}(\mathbf{L})}\tbinom{L_{r}+s_{r}-1}{L_{r}-1}\beta_{rk}^{s_{r}}\right\},
\]
where $\tbinom{L_{r}+s_{r}-1}{L_{r}-1}=\frac{\left(L_{r}+s_{r}-1\right)!}{\left(L_{r}-1\right)!s_{r}!}.$ The desired expansion is now given in the following lemma.

\begin{lemma}
$\Lambda^{\mathbf{L}} = \sum_{k \in \mathcal{D}(\mathbf{L})}\sum_{q=1}^{L_{k}}\frac{\gamma_{kq}(\mathbf{L})}{\lambda_{d + 1}^{N_i - q - L_{d + 1}}}\Lambda_{k}^{q}(t)\lambda_{d+1}^{N_{i}-q}$ if $\mathcal{D}(\mathbf{L}) \neq \emptyset.$ Further, this expansion is unique and holds $\forall t.$
\end{lemma}
\begin{proof}
See \cite{Gugan12}.
\end{proof}
Applying this expansion to each coefficient in \eqref{eq:afterColComTerms} and regrouping shows that it can be rewritten as
\begin{equation}
\label{eq:genCanPol}
f(\mathbf{x};t)=\sum_{k=1}^{d}\sum_{q=1}^{N_{i}}h_{kq}(\mathbf{x})\Lambda_{k}^{q}(t)\lambda_{d+1}^{N_{i}-q}-c(t),
\end{equation}
where $c(t)=\mu_i(t) - \lambda_{d+1}^{N_{i}}$ and
\begin{equation}
h_{kq}(\mathbf{x})=\sum_{\mathbf{L}\in\Delta_{d+1,N_{i}},\, L_{k}\ge q}\frac{\gamma_{kq}(\mathbf{L})}{\lambda_{d+1}^{N_{i}-q-L_{d+1}}}g(\mathbf{x};\mathbf{L}). \label{eq:elemPolWeiCombo}
\end{equation}

\noindent \emph{Step3-Eliminate dependence on $\tau$}: The advantage of \eqref{eq:genCanPol} is that apart from $c(t),$ the number of t-dependent coefficients equals $d\cdot N_{i},$ which is exactly the number of unknowns in the polynomial $f$. Further, as shown below, they are linearly independent.

\begin{lemma}
\label{lem:genTtauMatrix}
For $k \in [d\cdot N_i],$ let $b_k = \min \{ j \in [d]:j\cdot N_i \geq k \}.$ Then the matrix $T_{\tau}$, where for $j,k \in [d \cdot N_i]$
\begin{equation}
(T_{\tau})_{jk} = \Lambda_{b_k}^{k - (b_k - 1)\cdot N_i}(t_j) \lambda_{d + 1}^{b_k \cdot N_i - k}, \label{eq:coeffMat}
\end{equation}
is non-singular.
\end{lemma}
\begin{proof}
See \cite{Gugan12}.
\end{proof}
Now observe that if we let $\mathbf{c_{\tau}} \equiv (c(t_{1}),\ldots,c(t_{d\cdot N_{i}})),$ $\mathcal{E}_k(\mathbf{x}) \equiv (h_{k1}(\mathbf{x}), \ldots, h_{kN_{i}}(\mathbf{x}))$ and $\mathcal{E}(\mathbf{x})$ $\equiv (\mathcal{E}_1(\mathbf{x}),\ldots,$ $\mathcal{E}_d(\mathbf{x})),$ then \eqref{eq:tempPolSys} can be equivalently expressed as
\begin{equation}
\label{eq:genSymPolSys}
T_{\tau}\mathcal{E}(\mathbf{x})-\mathbf{c}_{\tau}=0.
\end{equation}
Premultiplying \eqref{eq:genSymPolSys} by $\left(T_{\tau}\right)^{-1},$ which now exists by Lemma~\ref{lem:genTtauMatrix}, we obtain
\begin{equation}
\label{eq:genSecSymPolSys}
\mathcal{E}(\mathbf{x})-\left(T_{\tau}\right)^{-1}\mathbf{c}_{\tau}=0.
\end{equation}
A crucial point to note now is that $\mathbf{w}\equiv(\mathbf{w}_{1},\ldots,\mathbf{w}_{N_{i}})$ is an obvious root of \eqref{eq:tempGenPoly} and hence of \eqref{eq:genSecSymPolSys}. This immediately implies that $\left(T_{\tau}\right)^{-1}\mathbf{c}_{\tau} = \mathcal{E}(\mathbf{w})$ and consequently \eqref{eq:genSecSymPolSys} can rewritten  as
\begin{equation}
H_{i}(\mathbf{x})\equiv\mathcal{E}(\mathbf{x})-\mathcal{E}(\mathbf{w})=0.\label{eq:genElePolSys}
\end{equation}
Note that \eqref{eq:genElePolSys} is devoid of any reference to the set $\tau$ and can be arrived at using any valid $\tau.$ Furthermore, because of the equivalence between \eqref{eq:tempPolSys} and \eqref{eq:genElePolSys}, any conclusion that we can draw for \eqref{eq:genElePolSys} must hold automatically for the system of \eqref{eq:tempPolSys}. Because of these reasons, we will henceforth refer to \eqref{eq:genElePolSys} as the EPS.

\begin{example}
Let $N_i = d = 2.$ Also, let $\lambda_1 = 5, \lambda_2 = 3$ and $\lambda_3 = 1.$ Then the map $\mathcal{E}$ described above is given by
\begin{equation}
\label{eq:egEPS}
\mathcal{E}(\mathbf{x})=\left(
\begin{array}{c}
x_{11} + x_{21} + 5(x_{11}x_{22} + x_{12}x_{21})\\

x_{11}x_{21}\\

x_{12} + x_{22} - 6(x_{11}x_{22} + x_{12}x_{21})\\

x_{12}x_{22}\\
\end{array}
\right).
\vspace{0.2cm}
\end{equation}

\end{example}

We next describe some special features of the EPS. Let $\mathbf{x}_{\sigma} := (\mathbf{x}_{\sigma(1)}, \ldots, \mathbf{x}_{\sigma(N_i)}),$ $\sigma \in S_{N_i},$  denote a permutation of the vectors $\mathbf{x}_1, \ldots, \mathbf{x}_{N_i}$ and $\pi_{\mathbf{x}}:=\{\mathbf{x}_{\sigma}: \sigma \in S_{N_i}\}.$

\begin{lemma}
\label{lem:EPSSymm}
$H_i(\mathbf{x}) = H_i(\mathbf{x}_\sigma),$ $\forall \sigma \in S_{N_i}.$ That is, the map $H_i$ is symmetric.
\end{lemma}
\begin{proof}
Observe that $H_i(\mathbf{x}) = (T_\tau)^{-1} F_{\tau}(\mathbf{x})$ and $F_\tau,$ as defined in \eqref{eq:tempPolSys}, is symmetric. The result thus follows.
\end{proof}

Next recall that if the complement of a solution set of a polynomial system is non-empty then it must be open dense in the Euclidean topology. This fact and  the above result help us now to show that the EPS is almost always well behaved.

\begin{lemma}
\label{lem:EPSWellBeh}
There exists an open dense set $\mathcal{R}$ of $\mathbb{R}^{d\cdot N_i}$ such that if $\mathbf{w} \in \mathcal{R}$ then the solution set of the EPS satisfies the following properties.
\begin{enumerate}
\item $\mathbf{w} \in V(H_i).$
\item If $\mathbf{x}^{*} \in V(H_i),$ then $\pi_{\mathbf{x}^{*}} \subset V(H_i).$ \label{prop:genSym}
\item $|V(H_i)| = k \times N_i!,$ where $k \in \mathbb{N}$ is independent of $\mathbf{w} \in \mathcal{R}.$ Further, each solution is non-singular.
\end{enumerate}
\end{lemma}
\begin{proof}
See \cite{Gugan12}.
\end{proof}

We henceforth assume that $\mathbf{w} \in \mathcal{R}.$ Property \ref{prop:genSym} above then suggests that it suffices to work with
\begin{equation}
\label{eq:redSolSet}
\mathcal{M}_i = \{\alpha \in \mathbb{C}^d : \exists \mathbf{x}^* \in V(H_i) \mbox{ with } \mathbf{x}_1^* = \alpha \}.
\end{equation}
Observe that $\mathcal{W}_i := \{\mathbf{w}_1, \ldots, \mathbf{w}_{N_i}\} \subset \mathcal{M}_i.$ A point to note here is that $\mathcal{I}_i := \mathcal{M}_i \backslash \mathcal{W}_i$ is not empty in general.

Our next objective is to develop the above theory for the case where for each $i \in [m],$ instead of the exact value of $M_{Y_i}(t),$ we have access only to the IID realizations $\{Y_{il}\}_{l \geq 1}$ of the random variable $Y_i.$ That is, for each $k \in [N_i],$ we have to use the sample average $\hat{M}_{Y_i}(t_k; L) = \left(\sum_{l = 1}^{L} \exp(-t_k Y_{il})\right)/L $ for an appropriately chosen large $L,$ $\hat{c}(t_k; L) = \hat{M}_{Y_i}(t_k)(\lambda_{d+1} + t_k)^{N_i} - \lambda_{d+1}^{N_i}$ and $\hat{\mathbf{c}}_{\tau, L} \equiv (\hat{c}(t_1; L), \ldots, \hat{c}(t_{d \cdot N_i}; L))$ as substitutes for each $M_{Y_i}(t_k),$ each $c(t_k)$ and $\mathbf{c}_\tau$ respectively. But even then note that the noisy or the perturbed version of the EPS
\begin{equation}
\label{eq:NoiGenElePolSys}
\hat{H}_i(\mathbf{x}) \equiv \mathcal{E}(\mathbf{x}) - (T_{\tau})^{-1}(\hat{\mathbf{c}}_{\tau, L}) = 0.
\end{equation}
is always well defined. More importantly, the perturbation is only in its constant term. As in Lemma~\ref{lem:EPSSymm}, it then follows that the map $\hat{H}_i$ is symmetric.

Next observe that since $\mathcal{R}$ is open (see Lemma~\ref{lem:EPSWellBeh}), there exists a small enough $\bar{\delta}_i > 0$ such that $B(\mathbf{w}; \bar{\delta}_i) \subset \mathcal{R}.$ Using the regularity of solutions of the EPS (see Property 3 of Lemma~\ref{lem:EPSWellBeh}), the inverse function theorem then gives us the following result.

\begin{lemma}
\label{lem:NoiEPSWellBeh}
Let $\delta \in (0, \bar{\delta}_i)$ be such that for any two distinct solutions in $V(H_i),$ say $\mathbf{x}^{*}$ and $\mathbf{y}^{*},$ $B(\mathbf{x}^{*}; \delta) \cap B(\mathbf{y}^{*}; \delta) = \emptyset.$ Then there exists an $\epsilon(\delta) > 0$ such that if $\mathbf{u} \in \mathbb{R}^{d \cdot N_i}$ and $||\mathbf{u} - \mathcal{E}(\mathbf{w})|| < \epsilon(\delta)$ then the solution set $V(\hat{H}_i)$ of the perturbed EPS $\mathcal{E}(\mathbf{x}) - \mathbf{u} = 0$ satisfies the following:
\begin{enumerate}
\item All roots in $V(\hat{H}_i)$ are regular points of the map $\mathcal{E}.$

\item For each $\mathbf{x}^{*} \in V(H_i),$ there is a unique $\mathbf{z}^{*} \in V(\hat{H}_i)$ such that $||\mathbf{x}^{*} - \mathbf{z}^{*}|| < \delta.$

\item For each $\mathbf{z}^{*} \in V(\hat{H}_i),$ there is a unique $\mathbf{x}^{*} \in V(H_i)$ such that $||\mathbf{x}^{*} - \mathbf{z}^{*}|| < \delta.$
\end{enumerate}
\end{lemma}
\begin{proof}
See \cite{Gugan12}.
\end{proof}

The above result, in simple words, states that if we can get hold of a close enough approximation of $\mathcal{E}(\mathbf{w}),$ say $\mathbf{u},$ then solving the perturbed EPS $\mathcal{E}(\mathbf{x}) - \mathbf{u} = 0$ is almost as good as solving the EPS of \eqref{eq:genElePolSys}. We now show how to acquire such an approximation of $\mathcal{E}(\mathbf{w}).$

\begin{lemma}
\label{lem:ExistL}
Let $\delta$ and $\epsilon(\delta)$ be as described in Lemma~\ref{lem:NoiEPSWellBeh}. Then for tolerable failure rate $\kappa > 0$ and the chosen set $\tau,$ $\exists L_{\tau, \delta, \kappa} \in \mathbb{N}$ such that if $L \geq L_{\tau, \delta, \kappa}$ then with probability greater than $1 - \kappa,$ we have $|| (T_\tau)^{-1} \hat{\mathbf{c}}_{\tau, L} - \mathcal{E}(\mathbf{w}) || < \epsilon(\delta).$
\end{lemma}
\begin{proof}
Note that $\exp(-t_k Y_{il}) \in [0,1]$ $\forall i, l$ and $k.$ The Hoeffding inequality (see \cite{Hoeffding62}) then shows that for any $\epsilon > 0$, $\Pr\{|\hat{M}_{Y_i}(t_k; L) - M_{Y_i}(t_k)| > \epsilon\} \leq \exp(-2 \epsilon^2 L).$ Since $\mathcal{E}(\mathbf{w}) = (T_\tau)^{-1}\mathbf{c}_\tau,$ the result is now immediate.
\end{proof}

Let us now fix a $L \geq L_{\tau, \delta, \kappa}$ and let $\mathcal{A}_i(\kappa)$ denote the event $|| (T_\tau)^{-1} \hat{\mathbf{c}}_{\tau, L} - \mathcal{E}(\mathbf{w}) || < \epsilon(\delta).$ Clearly, $\Pr\{\mathcal{A}^{c}_i(\kappa)\} \leq \kappa.$ Observe that when $\mathcal{A}_i(\kappa)$ is a success the solution set of \eqref{eq:NoiGenElePolSys}, with $L$ as chosen above, satisfies all properties given in Lemma~\ref{lem:NoiEPSWellBeh}. Because of the symmetry of the map $\hat{H}_i,$ as in \eqref{eq:redSolSet}, it again suffices to work with
\begin{equation}
\hat{\mathcal{M}}_i = \{\hat{\alpha} \in \mathbb{C}^d : \exists \mathbf{z}^{*} \in V(\hat{H}_i) \mbox{ with } \mathbf{z}_1^{*} = \alpha\}.
\end{equation}

We are now done discussing the EPS for an arbitrary $i \in [m].$ In summary, we have managed to obtain a set $\hat{\mathcal{M}}_i$ in which a close approximation of the weight vectors of random variables ${X_j}$ that add up to give $Y_i$ are present with high probability. The next subsection takes a unified view of the solution sets $\{\hat{\mathcal{M}}_i: i\in [m]\}$ to match the weight vectors to the corresponding random variables. But before that, we redefine $\mathcal{W}_i$ as $\{\mathbf{w}_j: j \in p_i\}.$ Accordingly, $\mathcal{M}_i, \hat{\mathcal{M}}_i, V(H_i)$ and $V(\hat{H}_i)$ are also redefined using notations of Section~\ref{sec:model}.

\subsection{Parameter Matching using 1-identifiability}
\label{subsec:parMatch}

We begin by giving a physical interpretation for the $1-$identifiability condition of the matrix $A.$ For this, let $\mathcal{G}_{j} := \left\{ i\in[m]:j\in p_{i}\right\} $ and $\mathcal{B}_{j} := [m] \backslash \mathcal{G}_{j}.$
\begin{lemma}
  \label{lem:linkAsPathIntersection}
  For a $1-$identifiable matrix $A$, each index $j\in[N]$ satisfies
  \[
  \{j\} = \bigcap_{g\in\mathcal{G}_{j}} p_{g} \cap \bigcap_{b\in\mathcal{B}_{j}} p_{b}^{c} =: \mathcal{D}_{j}.
  \]
\end{lemma}
\begin{proof}
By definition, $j \in \mathcal{D}_{j}.$ For converse, if $k\in\mathcal{D}_{j},$ $k \neq j,$ then columns $j$ and $k$ of $A$ are identical; contradicting its $1-$identifiability condition. Thus $\{j\}=\mathcal{D}_{j}.$
\end{proof}

An immediate result is the following.

\begin{corollary}
  \label{cor:uniqueValuePathValueIntersection}
  Suppose $A$ is a $1-$identifiable matrix. If the map $u: [N] \rightarrow X,$ where $X$ is an arbitrary set, is bijective and $\forall i \in [m],$ $v_{i} := \{u(j):j\in p_{i}\},$ then for each $j \in [N]$
  \[
  \{u(j)\} =  \bigcap_{g\in\mathcal{G}_{j}} v_{g} \cap \bigcap_{b\in\mathcal{B}_{j}} v_{b}^{c}
  \]
\end{corollary}

By reframing this, we get the following result.

\begin{theorem}
  \label{thm:idlCasParMch}
Suppose $A$ is a $1-$identifiable matrix. If the weight vectors $\mathbf{w}_1, \ldots, \mathbf{w}_N$ are pairwise distinct then the rule
  \begin{equation}
    \psi:j \rightarrow \bigcap_{g\in\mathcal{G}_{j}} \mathcal{W}_{g} \cap \bigcap_{b\in\mathcal{B}_{j}} \mathcal{W}_{b}^{c},
    \label{eq:idlCasAssRule}
  \end{equation}
satisfies $\psi(j) = \mathbf{w}_{j}$.
\end{theorem}

This result is where the complete potential of the $1-$ identifiability condition of $A$ is being truly taken advantage of. What this states is that if we had access to the collection of solution sets $\{\mathcal{W}_i: i \in [m]\}$ then using $\psi$ we would have been able to uniquely match the weight vectors to the random variables. But note that, at present, we have access only to the collection  $\{\mathcal{M}_i: i \in [m]\}$ in the ideal case and $\{\hat{\mathcal{M}}_i: i \in [m]\}$ in the perturbed case. Keeping this in mind, our goal now is to show that if $\forall i_1, i_2 \in [m],$ $i_1 \neq i_2,$
\begin{equation}
\mathcal{I}_{i_1} \cap \mathcal{M}_{i_2} = \emptyset, \label{eq:inValSolNotRel}
\end{equation}
a condition that always held in simulation experiments, then the rules (with minor modifications):
\begin{equation}
  \label{eq:idlCasGenAssRule}
	\psi:j \rightarrow \bigcap_{g\in\mathcal{G}_{j}} \mathcal{M}_{g} \cap \bigcap_{b\in\mathcal{B}_{j}} \mathcal{M}_{b}^{c}
\end{equation}
for the ideal case, and
\begin{equation}
  \label{eq:PerCasGenAssRule}
	\hat{\psi}:j \rightarrow \bigcap_{g\in\mathcal{G}_{j}} \hat{\mathcal{M}}_{g} \cap \bigcap_{b\in\mathcal{B}_{j}} \hat{\mathcal{M}}_{b}^{c}
\end{equation}
in the perturbed case, recover the correct weight vector associated to each random variable $X_j.$

We first discuss the ideal case. Let $\mathcal{S}:=\{j \in [N]: |\mathcal{G}_j| \geq 2\}.$ Because of \eqref{eq:inValSolNotRel} and Theorem~\ref{thm:idlCasParMch}, note that
\begin{enumerate}
\item If $j \in \mathcal{S},$  then $\psi(j) = \{\mathbf{w}_j\}.$

\item If $j \in \mathcal{S}^{c}, j \in p_{i^{*}}$ then $\psi(j) = \{\mathbf{w}_j\} \cup \mathcal{I}_{i^*}.$
\end{enumerate}
That is, \eqref{eq:idlCasGenAssRule} works perfectly fine when $j \in \mathcal{S}.$ The problem arises when $j \in \mathcal{S}^{c}$ as $\psi(j)$ does not give as output a unique vector. To correct this, fix $j \in \mathcal{S}^{c}.$ If $j \in p_{i^*}$ then let $\mathbf{v}^{sub} \equiv (\mathbf{w}_k : k \in p_{i^*} \backslash \{j\}).$ Because of $1-$identifiability, note that if $k \in p_{i^*}\backslash \{j\}$ then $k \in \mathcal{S}.$ From \eqref{eq:tempGenPoly} and \eqref{eq:tempPolSys}, it is also clear that $(\mathbf{v}^{sub}, \alpha) \in V(H_i)$ if and only if $\alpha = \mathbf{w}_j.$ This suggests that we need to match parameters in two stages. In stage 1, we use \eqref{eq:idlCasGenAssRule} to assign weight vectors to all those random variables $X_j$ such that $j \in \mathcal{S}.$ In stage 2, for each $j \in \mathcal{S}^{c},$ we identify $i^* \in [m]$ such that $j \in p_{i^*}.$  We then construct $\mathbf{v}^{sub}.$ We then assign to $j$ that unique $\alpha$ for which $(\mathbf{v}^{sub}, \alpha) \in V(H_{i^*}).$ Note that we are ignoring the trivial case where $|p_{i^*}| = 1.$ It is now clear that by using \eqref{eq:idlCasGenAssRule} with modifications as described above, at least for the ideal case, we can uniquely recover back for each random variable $X_j$ its corresponding weight vector $\mathbf{w}_j.$

We next handle the case of noisy measurements. Let $\mathcal{U}:= \cup_{i \in [m]} \mathcal{M}_i$ and $\hat{\mathcal{U}} := \cup_{i \in [m]} \hat{\mathcal{M}}_i.$ Observe that using \eqref{eq:PerCasGenAssRule} directly, with probability one, will satisfy $\hat{\psi}(j) = \emptyset$ for each $j \in [N].$ This happens because we are distinguishing across the solution sets the estimates obtained for a particular weight vector. Hence as a first step we need to define a relation $\sim$ on $\hat{\mathcal{U}}$ that associates these related elements. Recall from Lemmas~\ref{lem:NoiEPSWellBeh} and \ref{lem:ExistL} that the set $\hat{\mathcal{M}}_i$ can be constructed for any small enough choice of $\delta, \kappa > 0.$ With choice of $\delta$ that satisfies
\begin{equation}
\label{eq:delCond}
0 < 4\delta < \underset{\alpha, \beta \in \mathcal{U}}{\min} ||\alpha - \beta||,
\end{equation}
let us consider the event $\mathcal{A} : = \cap_{i \in [m]} \mathcal{A}_i(\kappa/m).$  Using a simple union bound, it follows that $\Pr\{\mathcal{A}^{c}\} \leq \kappa.$ Now suppose that the event $\mathcal{A}$ is a success. Then by \eqref{eq:delCond} and Lemma~\ref{lem:NoiEPSWellBeh}, the following observations follow trivially.

\begin{enumerate}
\item For each $i \in [m]$ and each $\alpha \in \mathcal{M}_i,$ there exists at least one $\hat{\alpha} \in \hat{\mathcal{M}}_i$ such that $||\hat{\alpha} - \alpha|| < \delta.$

\item For each $i \in [m]$ and each $\hat{\alpha} \in \hat{\mathcal{M}}_i,$ there exists precisely one $\alpha \in \mathcal{M}_i$ such that $||\hat{\alpha} - \alpha|| < \delta.$

\item Suppose for distinct elements $\alpha, \beta \in \mathcal{U},$ we have $\hat{\alpha}, \hat{\beta} \in \hat{\mathcal{U}}$ such that $||\hat{\alpha} - \alpha|| < \delta$ and $||\hat{\beta} - \beta|| < \delta.$ Then $||\hat{\alpha} - \hat{\beta}|| > 2\delta.$
\end{enumerate}

From these, it is clear that the relation $\sim$ on $\hat{\mathcal{U}}$ should be
\begin{equation}
\label{eq:Rel}
  \hat{\alpha} \sim \hat{\beta} \mbox{ iff } ||\hat{\alpha} - \hat{\beta}|| < 2\delta.
\end{equation}
It is also easy to see that, whenever the event $\mathcal{A}$ is a success, $\sim$ defines an equivalence relation on $\hat{\mathcal{U}}.$ For each $i \in [m],$ the obvious idea then is to replace each element of $\hat{\mathcal{M}}_i$ and its corresponding $d-$ dimensional component in $V(\hat{H}_i)$ with its equivalence class. It now follows that \eqref{eq:PerCasGenAssRule}, with modifications as was done for the ideal case, will satisfy
\begin{equation}
\hat{\psi}(j) = \{\hat{\alpha} \in \hat{\mathcal{U}}:||\hat{\alpha}- \mathbf{w}_j|| < \delta\}.
\end{equation}
This is obviously the best we could have done starting from the set $\{\hat{\mathcal{M}}_i: i \in [m]\}.$

We end this section by summarizing our complete method in an algorithmic fashion.  For each $i \in [m],$ let $\{Y_{il}\}_{l\geq 1}$ be the IID samples of $Y_i.$

\begin{algorithm} \textbf{\large Distribution tomography}
\label{alg:EstLinkDis}

\noindent \textit{\large Phase 1: Construct \& Solve the EPS.}

For each $i \in [m],$
\begin{enumerate}
\item Choose an arbitrary $\tau = \{t_1, \ldots, t_{d \cdot N_i}\}$ of distinct positive real numbers.
\item For a large enough $L\in\mathbb{N}$ and each $t_{j}\in\tau,$ set $\hat{M}_{Y_{i}}(t_{j})=\left(\sum_{l=1}^{L}\exp(-t_jY_{il})\right)/L,$ $\hat{\mu}_{i}(t_{j})=(\lambda_{d+1}+t_{j})^{N_{i}}\hat{M}_{Y_{i}}(t_{j})$ and $\hat{c}(t_{j})=\hat{\mu}_{i}(t_{j}) - \lambda_{d+1}^{N_{i}}.$  Using this, construct $\hat{c}_{\tau}\equiv(\hat{c}(t_{1}),\ldots,\hat{c}(t_{d\cdot N_{i}})).$
\item Solve $\mathcal{E}(\mathbf{x})-T_{\tau}^{-1}\hat{c}_{\tau}=0$ using any standard solver for polynomial systems.
\item Build $\hat{\mathcal{M}}_i = \{\alpha \in \mathbb{C}^d : \exists \mathbf{x}^* \in V(\hat{H}_i) \mbox{ with } \mathbf{x}_1^* = \alpha \}.$
\end{enumerate}

\noindent \textit{\large Phase 2: Parameter Matching}

\begin{enumerate}
\item Set $\hat{\mathcal{U}}:=\bigcup_{i\in[m]}\hat{\mathcal{M}}_{i}.$ Choose $\delta>0$ small enough and define the relation $\sim$ on $\mathcal{\hat{U}},$ where $\hat{\alpha} \sim \hat{\beta}$ if and only if $||\hat{\alpha} - \hat{\beta}||_{2} < 2\delta.$ If $\sim$ is not an equivalence relation then choose a smaller $\delta$ and repeat.
\item Construct the quotient set $\hat{\mathcal{U}}\backslash\sim.$ Replace all elements of each $\hat{\mathcal{M}_{i}},$ $V(\hat{H}_{i})$ with their equivalence class.
\item For each $j\in\mathcal{S},$ set \\$\hat{\psi}(j)=\left(\bigcap_{g\in\mathcal{G}_{j}}\mathcal{\hat{M}}_{g}\right)\cap\left(\bigcap_{b\in\mathcal{B}_{j}}\hat{\mathcal{M}}_{b}^{c}\right).$
\item For each $j\in\mathcal{S}^{c},$

\begin{enumerate}
\item Set $i^{*} = i\in[m]$ such that $j\in p_{i^{*}}.$
\item Construct $\mathbf{v}^{sub}\equiv(\psi(k):k\in p_{i^{*}}, k\ne j\}.$
\item Set $\psi(j)=\hat{\alpha}$ such that $(\mathbf{v}^{sub},\hat{\alpha})\in V(\hat{H}_{i}).$

\end{enumerate}
\end{enumerate}
\end{algorithm}

\section{Universality}
\label{sec:universality}
The crucial step in the method described above was to come up with, for each $i \in [m],$ a well behaved polynomial system, i.e., one that satisfies the properties of Lemma~\ref{lem:EPSWellBeh}, based solely on the samples of $Y_i.$ Once that was done, the ability to match parameters to the component random variables was only a consequence of the $1-$identifiability condition of the matrix $A.$ This suggests that it may be possible to develop similar schemes even in settings different to the ones assumed in Section~\ref{sec:model}. In fact, functions other than the MGF could also serve as blueprints for constructing the polynomial system. We discuss in brief few of these ideas in this section. Note that we are making a preference for polynomial systems for the sole reason that there exist computationally efficient algorithms, see for example \cite{Sommese05, Li97, Morgan89, Verschelde99}, to determine all its roots.

Consider the case where $\forall j \in [N],$ the distribution of $X_j$ is a finite mixture model given by
\begin{equation}
\label{eq:genMixDist}
F_j(u) = \sum_{k = 1}^{d_j + 1}w_{jk}\phi_{jk}(u),
\end{equation}
where $d_j \in \mathbb{N},$ $w_{j1}, \ldots, w_{j(d_l + 1)}$ denote mixing weights, i.e., $w_{jk} \geq 0$ and $\sum_{k = 1}^{d_k + 1}w_{jk} = 1,$ and $\{\phi_{jk}(u)\}$ are some basis functions, say Gaussian, uniform, etc. The MGF is thus given by
\begin{equation}
M_{X_j}(t) = \sum_{k = 1}^{d_j + 1}w_{jk}\int_{u = 0}^{\infty}\exp(-ut)d\phi_{jk}(u).
\end{equation}
Note now that if the basis functions $\{\phi_{jk}\}$ are completely known then the MGF of each $Y_i$ will again be a polynomial in the mixing weights, $\{w_{jk}\},$ similar in spirit to the relation of \eqref{eq:MGFRel}. As a result, the complete recipe of Section~\ref{sec:GenNT} can again be attempted to estimate the weight vectors of the random variables $X_j$ using only the IID samples of each $Y_i.$

In relation to \eqref{eq:Dist} or \eqref{eq:genMixDist}, observe next that $\forall n \in \mathbb{N},$ the $n^{th}$ moment of each $X_j$ is given by
\begin{equation}
\mathbb{E}(X_j^n) = \sum_{k = 1}^{d_j + 1}w_{jk}\int_{u = 0}^{\infty}u^n d\phi_{jk}(u).
\end{equation}
Hence, the $n^{th}$ moment of $Y_i$ is again a polynomial in the unknown weights. This suggests that, instead of the MGF, one could use the estimates of the moments of $Y_i$ to come up with an alternative polynomial system and consequently solve for the distribution of each $X_j.$

Moving away from the models of \eqref{eq:Dist} and \eqref{eq:genMixDist}, suppose that for each $j \in [N],$ $X_j \sim \exp(m_j).$ Assume that each mean $m_j < \infty$ and that  $m_{j_{1}} \neq m_{j_{2}}$ when $j_{1}\neq j_{2}.$ We claim that the basic idea of our method can be used here to estimate $m_1, \ldots, m_N$ and hence the complete distribution of each $X_j$ using only the samples of $Y_i.$ As the steps are quite similar when either i) we know $M_{Y_i}(t)$ for each $i \in [m]$ and every valid $t$ and ii) we have access only to the IID samples $\{Y_{il}\}_{l \geq 1}$ for each $i \in [m]$, we take up only the first case.

Fix $i \in [m]$ and let $p_i := \{j \in [N]: a_{ij} = 1\}.$ To simplify notations, let us relabel the random variables $\{X_j: j \in p_i\}$ that add up to give $Y_i$ as $X_1, \ldots, X_{N_i},$ where $N_i = |p_i|.$ Observe that the MGF of $Y_i,$ after inversion, satisfies
\begin{equation}
  \prod_{j=1}^{N_{i}} (1 + t m_{j}) = 1/M_{Y_{i}}(t).
  \label{eq:momGenRel}
\end{equation}
Using \eqref{eq:momGenRel}, we can then define the canonical polynomial
\begin{equation}
  f(\mathbf{x};t):=\prod_{j=1}^{N_{i}}(1+tx_{j})-c(t), \label{eq:canPol}
\end{equation}
where $\mathbf{x} \equiv (x_{1},\ldots,x_{N_{i}})$ and $c(t) = 1/M_{Y_{i}}(t).$ Now choose an arbitrary set $\tau = \{t_{1},\ldots,t_{N_{i}}\} \subset \mathbb{R}_{++}$ consisting of distinct numbers and define
\begin{equation}
  F_{\tau}( \mathbf{x}) \equiv (f_{1}(\mathbf{x}), \ldots,  f_{N_{i}}(\mathbf{x})) = 0, \label{eq:IntmdPolySys}
\end{equation}
where $f_{k}(\mathbf{x})=f(\mathbf{x};t_{k}).$ We emphasize that this system is square of size $N_{i},$ depends on the choice of subset $\tau$ and each polynomial $f_{k}$ is symmetric with respect to the variables $x_{1,}\ldots,x_{N_{i}}.$ In fact, if we let $\mathbf{c}_{\tau}\equiv(c(t_{1}),\ldots,c(t_{N_{i}}))$ and $\mathcal{E} (\mathbf{x}) \equiv (e_{1}(\mathbf{x}), \ldots, e_{N_{i}} (\mathbf{x})),$ where $e_{k}(\mathbf{x})=\sum_{1\leq j_{1}<j_{2}<\ldots<j_{k}\leq N_{i}}x_{j_{1}}\cdots x_{j_{k}}$ denotes the $k^{th}$ elementary symmetric polynomial in the $N_{i}$ variables $x_{1,}\ldots,x_{N_{i}},$ we can rewrite \eqref{eq:IntmdPolySys} as
\begin{equation}
  T_{\tau} \mathcal{E}(\mathbf{x}) - (\mathbf{c}_{\tau} - \mathbf{1}) = 0. \label{eq:firstSymPolSys}
\end{equation}
Here $T_{\tau}$ denotes a Vandermonde matrix of order $N_{i}$ in $t_{1},\ldots,t_{N_{i}}$ with $(T_{\tau})_{jk}=t_{j}^{k}.$ Its determinant, given by $\det(T_{\tau}) = \left(\prod_{j=1}^{n} t_{j} \right) \prod_{j>i}(t_{j}-t_{i}),$ is clearly non-zero. Premultiplying \eqref{eq:firstSymPolSys} by $T_{\tau}^{-1},$ we obtain
\begin{equation}
  \mathcal{E}(\mathbf{x}) - T_{\tau}^{-1}(\mathbf{c}_{\tau} -   \mathbf{1}) = 0. \label{eq:secSymPolSys}
\end{equation}
Observe now that the vector $\mathbf{m} \equiv (m_{1},\ldots,m_{N_{i}})$ is a natural root of \eqref{eq:IntmdPolySys} and hence of \eqref{eq:secSymPolSys}. Hence $T_{\tau}^{-1}(\mathbf{c}_{\tau} - \mathbf{1}) = \mathcal{E}(\mathbf{m}).$  The EPS for this case can thus be written as
\begin{equation}
  H_{i}(\mathbf{x}) \equiv \mathcal{E}(\mathbf{x}) - \mathcal{E}(\mathbf{m}) = 0. \label{eq:charPolSys}
\end{equation}

We next discuss the properties of this EPS, or more specifically, its solution set. For this, let $V(H_{i}):=\{\mathbf{x} \in \mathbb{C}^{N_i} : H_{i}(\mathbf{x}) = 0\}.$
\begin{lemma}
	\label{lem:EPSSolSet}
	$V(H_{i})=\pi_{\mathbf{m}}:= \{\sigma(\mathbf{m}): \sigma \in S_{N_{i}}\}.$
\end{lemma}
\begin{proof} This follows directly from \eqref{eq:charPolSys}.
\end{proof}

\begin{lemma}
	\label{lem:EPSRegSol}
	For every $\mathbf{x}^{*}\in V(H_{i}),$ $\det(\dot{\mathcal{E}}(\mathbf{x^{*}})) \neq 0.$ \label{prop:regSol}
\end{lemma}
\begin{proof} This follows from Lemma (\ref{lem:EPSSolSet}) and the fact that $\det(\dot{\mathcal{E}}(\mathbf{x}))=\prod_{1\leq j<k\leq N_{i}}(x_{j}-x_{k}).$
\end{proof}

Because of Lemma~\ref{lem:EPSSolSet}, it suffices to work with only the first components of the roots. Hence we define
\begin{equation}
  \mathcal{M}_{i}:=\{\alpha^{*}\in\mathbb{C}: \exists \mathbf{x}^{*} \in V(H_i) \mbox{ with } \mathbf{x}^*_{1} = \alpha \},
\end{equation}
which in this case is equivalent to the set $\{m_1, \ldots, m_{N_i} \}.$ Reverting back to global notations, note that
\begin{equation}
  \mathcal{M}_{i}=\{m_j : j \in p_i\}. \label{eq:finPathSolSet}
\end{equation}

Since $i$ was arbitrary, we can repeat the above procedure to obtain the collection of solution sets $\{\mathcal{M}_i: i \in [m]\}.$ Arguing as in Theorem~\ref{thm:idlCasParMch}, it is now follows that if $A$ is $1-$identifiable then the rule
\begin{equation}
\psi(j) = \bigcap_{g \in \mathcal{G}_j}\mathcal{M}_g \cap \bigcap_{b \in \mathcal{B}_j} \mathcal{M}_b^{c},
\end{equation}
where $\mathcal{G}_j = \{i \in [m]: j \in p_i\}$ and $\mathcal{B}_j = [m]\backslash\mathcal{G}_j,$ satisfies the relation $\psi(j) = m_j.$  That is, having obtained the sets $\{\mathcal{M}_i : i \in [m]\},$ one can use $\psi$ to match the parameters to the corresponding random variables.

This clearly demonstrates that even if a transformation of the MGF is a polynomial in the parameters to be estimated, our method may be applicable.

\addtolength{\textheight}{-3.5cm}   % This command serves to balance the column lengths
                                  % on the last page of the document manually. It shortens
                                  % the textheight of the last page by a suitable amount.
                                  % This command does not take effect until the next page
                                  % so it should come on the page before the last. Make
                                  % sure that you do not shorten the textheight too much.

%%%%%%%%%%%%%%%%%%%%%%%%%%%%%%%%%%%%%%%%%%%%%%%%%%%%%%%%%%%%%%%%%%%%%%%%%%%%%%%%

\section{Experimental Results}
\label{sec:ExptResults}

\begin{figure}[ht!]
\begin{center}
\subfigure[Tree topology]
{
	\label{fig:tree}
  \includegraphics[width=0.4\columnwidth]{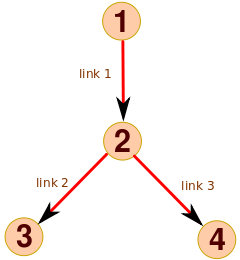}
}
\subfigure[General topology]
{
	\label{fig:general}
  \includegraphics[width=0.4\columnwidth]{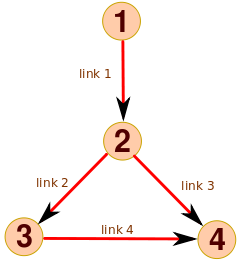}
}
\end{center} \label{fig:simNet1}
\caption{Network topologies in simulation experiments.}
\end{figure}
In order to verify performance, we conducted matlab based simulations using a network with i) tree topology, Fig.~\ref{fig:tree} and ii) general topology, Fig.~\ref{fig:general}. 

\subsection{Network with Tree Topology}
\label{subsec:treeTopology}

We work here with the network of Figure ~\ref{fig:tree}.
\begin{experiment}
\label{exp:expt1}
Node $1$ is the source node, while nodes $3$ and $4$ act as sink. The packet delay across link $j,$ denoted $X_j,$ is a hyperexponential random variable---a special case of the \eqref{eq:Dist}. The count of exponential stages in each link distribution equals three. That is, $d = 2.$ The corresponding exponential stage parameters $\lambda_1, \lambda_2$ and $\lambda_3$ are taken to be $5, 3$ and $1$ respectively. For each link, the weight associated with each exponential stage is given in first three columns of Table~\ref{tab:expt1}.

\begin{table}[h!]
\caption{\label{tab:expt1}Actual and Estimated weights for Expt.~\ref{exp:expt1}}
\centering
\begin{tabular}{| c | c  c  c | c  c  c |}
\hline
Link & $w_{j1}$ & $w_{j2}$ & $w_{j3}$ & $\hat{w}_{j1}$ & $\hat{w}_{j2}$ & $\hat{w}_{j3}$\\
\hline 
1 & 0.17  & 0.80 & 0.03 &  0.15  & 0.82 & 0.02\\
 
2 & 0.13  & 0.47 & 0.40 &  0.15  & 0.46 & 0.39\\

3 & 0.80  & 0.15 & 0.05 &  0.79  & 0.15 & 0.06\\
\hline 
\end{tabular}
\vspace{0.1cm}
\end{table}

Let $p_1$ be the path connecting the nodes $1, 2$ and $3.$ Similarly, let $p_2$ be the path connecting the nodes $1, 2$ and $4.$ Let $Y_1$ and $Y_2$ denote respectively the end-to-end delay across each of these paths. If we let $Y \equiv (Y_1, Y_2)$ and $X \equiv (X_1, X_2, X_3)$ then it follows that $Y = AX,$ where 
\[
A = \left (
\begin{array}{ccc}
1 & 1 & 0\\

1 & 0 & 1\\
\end{array}
 \right).
\]
We are now in the framework of Section~\ref{sec:model}.

We first focus on path $p_1.$ Observe that its EPS is given by the map of \eqref{eq:egEPS}. We collect now a million samples of its end-to-end delay. Choosing an arbitrary set $\tau = \{1.9857, 2.3782, 0.3581, 8.8619\},$ we run the first phase of Algorithm~\ref{alg:EstLinkDis} to obtain $\hat{\mathcal{M}}_1.$ This set along with its ideal counterpart is given in Table~\ref{tab:expt11}.

\begin{table}[h!]
\caption{\label{tab:expt11} Solution set of the EPS for path $p_1$ in Expt.~\ref{exp:expt1}}
\centering
\begin{tabular}{| c | c  c |}
\hline
sol-ID & $\mathcal{M}_1$ & $\hat{\mathcal{M}}_1$\\
\hline 
1 & (0.1300, 0.4700)  & (0.1542, 0.4558)\\
 
2 & (0.1700, 0.8000)  & (0.1292, 0.8356)\\

3 & (3.8304, -2.8410)  & (3.8525, -2.8646)\\

4 & (0.1933, 0.7768)  & (0.2260, 0.7394)\\

5 & (0.1143, 0.4840)  & (0.0882, 0.5152)\\

6 & (0.0058, -0.1323)  & (0.0052, -0.1330)\\
\hline 
\end{tabular}
\vspace{0.1cm}
\end{table}

Similarly, by probing the path $p_2$ with another million samples and with $\tau = \{0.0842, 0.0870, 0.0305, 0.0344\},$ we determine $\hat{\mathcal{M}}_2.$ The sets $\mathcal{M}_2$ and $\hat{\mathcal{M}}_2$ are given in Table~\ref{tab:expt12}.

\begin{table}[h!]
\caption{\label{tab:expt12} Solution set of the EPS for path $p_2$ in Expt.~\ref{exp:expt1}}
\centering
\begin{tabular}{| c | c  c |}
\hline
sol-ID & $\mathcal{M}_2$ & $\hat{\mathcal{M}}_2$\\
\hline 
1 & (0.8000, 0.1500)  & (0.7933, 0.1459)\\
 
2 & (0.1660, 0.7775)  & (0.1720, 0.8095)\\

3 & (5.5623, -4.5638)  & (5.5573, -4.5584)\\

4 & (0.1700, 0.8000)  & (0.1645, 0.7669)\\

5 & (0.8191, 0.1543)  & (0.8296, 0.1540)\\

6 & (0.0245, -0.0263)  & (0.0246, -0.0259)\\
\hline 
\end{tabular}
\vspace{0.1cm}
\end{table}

To match the weight vectors to corresponding links, firstly observe that the minimum distance between $\hat{\mathcal{M}}_1$ and $\hat{\mathcal{M}}_2$ is 0.0502. Based on this, we choose $\delta = 0.03$ and run the second phase of Algorithm~\ref{alg:EstLinkDis}. The obtained results, after rounding to two significant digits, are given in the second half of Table~\ref{tab:expt1}. Note that the weights obtained for the first link are determined by taking a simple average of the solutions obtained from the two different paths. The norm of the error vector equals $0.0443.$
\end{experiment}

\begin{experiment}
\label{exp:expt2}
Keeping other things unchanged as in the setup of experiment~\ref{exp:expt1}, we consider here four exponential stages in the distribution of each $X_j.$ The exponential stage parameters $\lambda_1, \lambda_2, \lambda_3$ and $\lambda_4$ equal $5, 4, 0.005$ and $1$ respectively. The corresponding weights are given in Table~\ref{tab:expt2}. But observe that the weights of the third stage is negligible for all three links. We can thus ignore its presence. The results obtained after running Algorithm~\ref{alg:EstLinkDis} are given in the second half of Table~\ref{tab:expt2}. The norm of the error vector equals $0.1843.$

\begin{table}[h!]
\caption{\label{tab:expt2}Actual and Estimated weights for Expt.~\ref{exp:expt2}}
\begin{center}
\begin{tabular}{|c|c c c c|c c c c|}
\hline 
Link & $w_{j1}$ & $w_{j2}$ & $w_{j3}$ & $w_{j4}$ & $\hat{w}_{j1}$ & $\hat{w}_{j2}$ & $\hat{w}_{j3}$ & $\hat{w}_{j4}$\tabularnewline
\hline 
1 & 0.71  & 0.20 & 0.0010 & 0.08 & 0.77  & 0.20 & 0 & 0.03\\

2 & 0.41  & 0.17 & 0.0015 & 0.41 & 0.38  & 0.18 & 0 & 0.44\\

3 & 0.15  & 0.80 & 0.0002 & 0.04 & 0.12  & 0.70 & 0 & 0.18\\
\hline 
\end{tabular}
\end{center}
\vspace{0.1cm}
\end{table}

\end{experiment}

\subsection{Network with General Topology}
\label{subsec:genTopology}

We deal here with the network of Figure~\ref{fig:general}.
\begin{experiment}
\label{exp:expt3}
Nodes $1$ and $2$ act as source while nodes $3$ and $4$ act as sink. We consider here three paths. Path $p_1$ connects the nodes $1, 2$ and $3,$ path $p_2$ connects the nodes $1, 2$ and $4,$ while path $p_3$ connects the nodes $2, 3$ and $4.$ Let $X \equiv (X_1, X_2, X_3, X_4),$ where $X_j$ is the packet delay across link $j.$ Also, let $Y \equiv (Y_1, Y_2, Y_3),$ where $Y_i$ denotes the end-to-end delay across path $p_i.$ They are related by $Y = AX,$ where 
\[
A = \left(
\begin{array}{cccc}
1 & 1 & 0 & 0\\
1 & 0 & 1 & 0\\
0 & 1 & 0 & 1\\
\end{array}
\right).
\]

As in Experiment~\ref{exp:expt1}, the random variable $X_j$ is hyperexponentially distributed with values of $d, \lambda_1, \lambda_2, \lambda_3$ kept exactly the same. By choosing again a million probe packets for each path, we run Algorithm~\ref{alg:EstLinkDis}. The actual and estimated weights are shown in Table~\ref{tab:expt3}.

\begin{table}[h!]
\caption{\label{tab:expt3}Actual and Estimated weights for Expt.~\ref{exp:expt3}}
\begin{center}
\begin{tabular}{|c|c c c| c c c|}
\hline 
Link & $w_{j1}$ & $w_{j2}$ & $w_{j3}$ & $\hat{w}_{j1}$ & $\hat{w}_{j2}$ & $\hat{w}_{j3}$\tabularnewline
\hline 
1 & 0.34  & 0.26 & 0.40 & 0.34  & 0.24 & 0.42\\

2 & 0.46  & 0.49 & 0.05 & 0.45  & 0.50 & 0.05\\

3 & 0.12  & 0.65 & 0.23 & 0.11  & 0.68 & 0.21\\

4 & 0.71  & 0.19 & 0.10 & 0.69  & 0.18 & 0.13\\
\hline 
\end{tabular}
\end{center}
\vspace{0.1cm}
\end{table}
\end{experiment}

\section{Discussion}
\label{sec:ConclFutW}
This paper took advantage of the properties of polynomial systems to develop a novel algorithm for the GNT problem. For any arbitrary $1-$identifiable matrix $A,$ it demonstrated successfully how to accurately estimate the distribution of the random vector $X,$ with mutually independent components, using only IID samples of the components of the random vector $Y = AX.$ Translating to network terminology, this means that one can now address the tomography problem even for networks with arbitrary topologies using only pure unicast probe packet measurements. The fact that we need only the IID samples of the components of $Y$ shows that the processes to acquire these samples across different paths can be asynchronous. Another nice feature of this approach is that it can estimate the unknown link level performance parameters even when no prior information is available about the same.

\section*{Acknowledgement}
I would like to thank my advisors, Prof. V.~Borkar and Prof. D.~Manjunath, for guiding me right through this work. I would also like to thank C.~Wampler, A.~Sommese, R.~Gandhi, S.~Gurjar and M.~Gopalkrishnan for helping me understand several results from algebraic geometry.

\bibliographystyle{IEEE}
\bibliography{tomography}

\end{document}